\documentclass[11pt,twoside]{article}

\usepackage{a4}

\usepackage{amssymb,amsmath,amsthm,latexsym}
\usepackage{amsfonts}
  \usepackage{amsfonts}
\usepackage{graphicx}

\usepackage{mathtools,amsthm}
\newtheoremstyle{case}{}{}{}{}{}{:}{ }{}

\usepackage{hyperref}
\usepackage{amsmath, amsfonts}
\usepackage{amssymb, graphicx}
\usepackage{amscd}
\usepackage{textcomp}
\usepackage{palatino}
\usepackage{mathtools}
\usepackage{enumitem}
\usepackage[dvipsnames]{xcolor}

\newtheorem{theorem}{Theorem}[section]

\newtheorem{definition}[theorem]{Definition}
\newtheorem{example}[theorem]{Example}
\newtheorem{lemma} [theorem]{Lemma}

\newtheorem{remark}[theorem]{Remark}

\numberwithin{subcase}{case}

\vspace{5cm}

\begin{document}
  
  \label{'ubf'}  
\setcounter{page}{1}                                 %Put here the starting page number
\markboth {\hspace*{-9mm} \centerline{\footnotesize \sc
Semi-Involutory MDS matrices }
               }
              { \centerline                           {\footnotesize \sc  
 T. Chatterjee and A. Laha
 } \hspace*{-9mm}              
    }

\vspace*{-2cm}

\begin{center}
{ 
       { \textbf {  A Characterization of Semi-Involutory MDS Matrices
    % Put the title of the paper here
                               }
       }
\\

\medskip
{\sc Tapas Chatterjee }\\
{\footnotesize Indian Institute of Technology Ropar, Punjab, India.
}\\
{\footnotesize e-mail: {\it tapasc@iitrpr.ac.in}}
\medskip

{\sc Ayantika Laha }\\
{\footnotesize Indian Institute of Technology Ropar, Punjab, India.
}\\
{\footnotesize e-mail: {\it 2018maz0008@iitrpr.ac.in}}
\medskip
}
\end{center}

\thispagestyle{empty} 
\vspace{-.4cm}

\hrulefill

\begin{abstract}  
In symmetric cryptography, maximum distance separable (MDS) matrices with computationally simple inverses have wide applications. Many block ciphers like AES, SQUARE, SHARK, and hash functions like PHOTON use an MDS matrix in the diffusion layer. In this article, we first characterize all $3 \times 3$ irreducible semi-involutory matrices over the finite field of characteristic $2$. Using this matrix characterization, we provide a necessary and sufficient condition to construct $3 \times 3$ MDS semi-involutory matrices using only their diagonal entries and the entries of an associated diagonal matrix. Finally, we count the number of $3 \times 3$ semi-involutory MDS  matrices over any finite field of characteristic $2$.

\end{abstract}
\hrulefill



\vspace{-.37cm}

\section{\bf Introduction}

In the design of symmetric-key cryptography, the concept of ``confusion and diffusion'' was introduced by Claude Shannon in his seminal paper ``Communication Theory of Secrecy System'' \cite{CS}. The confusion layer hides the relationship between the key and the ciphertext while the goal of the diffusion layer is to conceal the relationship between the ciphertext and the plain text.
 %Nonlinear functions like S-boxes and Boolean functions are used to construct the confusion layer of Substitution Permutation network (SPN) based block cipher. 
 Perfect diffusion can be achieved by either multipermutations \cite{vau} or by using maximum distance separable (MDS) matrices. An MDS matrix offers perfect diffusion because of its maximum branch number. Due to this, MDS matrices play an important role in the security against differential and linear attacks in the design of block ciphers and hash functions. Many modern age block ciphers such as AES \cite{DR}, Twofish \cite{SKWWHF}, SQUARE \cite{DKR}, SHARK \cite{RDPB} etc., and hash functions like Whirlpool \cite{BR}, PHOTON \cite{GPP} rely MDS matrices for enhanced security.

It is known that not all MDS matrices are suitable for the diffusion layer of lightweight cryptography. In this context, the construction of an MDS matrix with efficient implementation is a well studied problem.  
In \cite{YTH}, Youssef {\it{et al.}} introduced the idea of involutory linear transformation in the construction of a substitution permutation network (SPN) based encryption scheme. Motivated by this, the authors in \cite{YMT} studied two construction methods for involution linear transformations based on MDS codes. In the first, they constructed an $n \times n$ involutory matrix from an arbitrary non-singular matrix of order $ \frac{n}{2} \times \frac{n}{2}$ over the finite field $\mathbb{F}_{2^m}$. The other method used the Cauchy matrix to construct an involutory matrix over the finite field. Subsequently, in $2012$, Sajadieh {\it{et al.}} constructed involutory MDS matrices from Vandermonde matrices over $\mathbb{F}_{2^m}$ \cite{SDMO}. In $2013$, Gupta and Ray \cite{GR} introduced four distinct methods for constructing MDS matrices using the Cauchy matrix. They also provided construction of an involutory MDS matrix over the finite field $\mathbb{F}_{2^m}$ when the order of the matrix is a power of $2$. Additionally, they gave the construction of an involutory Hadamard MDS matrix using Vandermonde based techniques in \cite{GR}. Inspired by the use of circulant MDS matrix in the diffusion layer of AES,~ Gupta {\it{et al.}} \cite{GR1} studied the circulant involutory matrices over $\mathbb{F}_{2^m}$ and proved the non-existence of involutory circulant matrices of both even and odd order. In $2016$, Liu and Sim \cite{LS} provided examples of left circulant MDS matrices of odd order with involutory property.  

 In $2007$, Barreto {\it{et al.}} designed a block cipher named Curupira \cite{BS}. Curupia is a specially designed block cipher for  platforms where power consumption and processing time are very constrained
resources, such as sensor and mobile networks or systems heavily reliant on tokens or smart cards. The involutory MDS matrix used in the block cipher is constructed in a particularly intriguing manner which is the following:\\ Consider the matrix $D=I+aA+bB, a,b \in \mathbb{F}_{2^m}$ with $A=\begin{bmatrix}
1 &1 &1\\
0& 0 & 0\\
1 &1& 1\\
\end{bmatrix}$ and $B=\begin{bmatrix}
0 &0 &0\\
1&1 & 1\\
1 &1& 1\\
\end{bmatrix}$. Here $A^2=B^2=O$ and $AB=BA$. Note that, $I$ is the identity matrix and $O$ is the zero matrix of order $3 \times 3$. Observe that the matrix $D$ satisfies $D^2=I$ over the finite field $\mathbb{F}_{2^m}$. Additionally, $D$ is MDS if and only if $a \neq 0,1; b \neq 0,1,a,a+1.$ Therefore $D=I+2A+4B$ is an MDS involutory matrix over finite fields of characteristic $2$ which is used in Curupia.

%: $$\begin{bmatrix}
%3 &2 &2\\
%4& 5 & 4\\
%6 & 6& 7\\
%\end{bmatrix}$$ over the finite field $\mathbb{F}_{2^8}$. It is worth noting that this matrix is involutory. 

Inspired by the use of order $3$ matrix in diffusion layer, in \cite{GSARC}, G\"{u}zel {\it{et al.}} introduced a general format to construct all $3 \times 3$ involutory MDS  matrices over $\mathbb{F}_{2^m}$.  They also counted the total number of such matrices and proved that the number of $1'$s in a $3 \times 3$ involutory MDS  matrix is at most $3$. Additionally, they provided experimental results regarding $3 \times 3$ involutory MDS  matrices with low implementation cost over the finite fields $\mathbb{F}_{2^3}, \mathbb{F}_{2^4}$ and $\mathbb{F}_{2^8}$. Many authors \cite{GR,SDMO} studied involutory MDS matrices due to their efficiency in the implementation for the decryption layer of block ciphers. Recently, in $2021$, Cheon {\it{et al.}} \cite{CCK} introduced semi-involutory matrices  as a broader generalization of the involutory property, thereby expanding the range of matrix constructions. Subsequently, in $2023$, Chatterjee {\it{et al.}} proved that some Cauchy based MDS constructions given by Gupta and Ray \cite{GR} are semi-involutory \cite{TAS}. They introduced a new construction of the MDS semi-involutory matrix using the Cauchy matrix over the finite field of characteristic $p >2$. Moreover, they proved that $3 \times 3$ semi-involutory matrices with all non-zero entries are MDS.
As a natural generalization of this result, in this article, we study whether it is possible to construct a representative matrix for these MDS semi-involutory matrices of order $3 \times 3$. Such direct constructions are particularly important because they significantly reduce the search space, making the design process for MDS matrices more efficient.
Moreover, in \cite{TA1, TA2, TA3}, the authors also studied circulant and cyclic matrices with semi-involutory and semi-orthogonal properties over finite fields.
%and characterized $4 \times 4$ semi-involutory matrices under certain restrictions on the entries of the matrix.
Many authors have continued the search of MDS matrices from finite fields to rings and modules. In $1995$, Zain and Rajan defined MDS codes over cyclic groups \cite{ZR} and Dong {\it{et al.}} characterized MDS codes over elementary Abelian groups \cite{DCG}. By considering a finite Abelian group as a torsion module over a PID, Chatterjee {\it{et al.}} proved some non-existence results of MDS matrices in $2022$ \cite{TAS1}.

\subsection{Contribution} 
In this article, we extend the previous results as in \cite{GSARC} for a bigger class of semi-involutory matrices and provide the general structure for $3 \times 3$ semi-involutory matrices in \S~\ref{general structures}. An advantage of considering $3 \times 3$ semi-involutory matrices instead of involutory ones is that they are MDS when all entries are non-zero over a finite field. Consequently, we present a necessary and sufficient condition for the general structure of $3 \times 3$ semi-involutory matrices to be MDS based on the aforementioned condition. In \S~\ref{cardinality},  we count the total number of $3 \times 3$ semi-involutory MDS matrices over a finite field of characteristic $2$ using a set of $6$-tuple. 

\section{Organization of the paper}
The organization of the paper is as follows. In section \ref{preli}  we provide definitions and preliminaries. In section \ref{general structures}, we prove the general structure of $3 \times 3$ semi-involutory matrices over a finite field of characteristic $2$. In section \ref{cardinality}, we determine the total number of $3 \times 3$ semi-involutory MDS matrices using the results from the previous section. Finally, in Section \ref{conclusion}, we conclude the paper.

\section{\bf Preliminaries}\label{preli}
Let $\mathbb{F}_q$ denote the finite field with $q=p^m$ elements where $p$ is a prime number and $m$ is a positive integer. $\mathbb{F}_q^*$ denotes the set of all non-zero elements of $\mathbb{F}_q$. Let $\mathcal{C}$ be an $[n,k,d]$ linear error correcting code over the finite field $\mathbb{F}_q$. Then $\mathcal{C}$ is a subspace of $\mathbb{F}_q^n$ of dimension $k$ such that the Hamming distance between any two vectors in $\mathcal{C}$ is at least $d$.
 %By singleton bound we know that $d \leq n-k+1$. 
 The code $\mathcal{C}$ is a maximum distance separable (MDS) code if the Singleton bound is attained, i.e., $d = n-k+1$. The generator matrix $G$ of an MDS code $\mathcal{C}$ is a $k \times n$ matrix such that any set of $k$ columns of $G$ is  linearly independent. The standard form of $G$ is $[I|A]$, where $I$  is a $k \times k$ identity matrix and $A$ is a $k \times (n-k)$ matrix. From \cite{MS}, another definition of MDS code is the following.
\begin{definition}
An $[n,k,d]$ code $\mathcal{C}$ with the generator matrix $G=[I|A]$, where $A$ is a $k \times (n-k)$ matrix, is MDS if and only if every $i \times i$ submatrix of $A$ is non-singular, $i=1,2,\cdots,\text{min}(k,n-k)$.
\end{definition}
 This definition of MDS code gives the following characterization of an MDS matrix.

\begin{definition}
A square matrix $A$ is said to be MDS if every square submatrix of $A$ is non-singular.
\end{definition}

The construction of an MDS matrix with an easily implementable inverse matrix has received a lot of attention due to the use of the inverse matrix in decryption layer of block cipher. Since an involutory matrix has the same implementation circuit with its inverse, involutory MDS matrices are beneficial when the decryption process is required.

\begin{definition}
A square matrix $A$ is said to be involutory if $A^2=I$.
\end{definition}
An example of MDS matrix with involutory property is the following:
\begin{example}
Consider the finite field $\mathbb{F}_{2^4}$ with generating polynomial $x^4+x+1$. Let $\alpha$ be a primitive element of this field. Consider the matrix $$M=\begin{bmatrix}
\alpha^3 & \alpha^3+1 & \alpha^3+1\\
\alpha^3+\alpha^2+\alpha & \alpha^3+\alpha^2+\alpha+1 & \alpha^3+\alpha^2+\alpha\\
\alpha^3+\alpha+1 & \alpha^3+\alpha+1 & \alpha^3+\alpha+1
\end{bmatrix}.$$
This matrix satisfy $M^2=I$ and also an MDS matrix. 
\end{example}

An exhaustive search for an involutory MDS matrix can be time-consuming, particularly when dealing with a large field size. 
As a result, determining the total number of involutory MDS matrices over a finite field becomes an interesting problem. In \cite{GSARC}, G\"{u}zel {\it{et al.}} 
proved that there exist $(2^m-1)^2(2^m-2)(2^m-4)$ involutory MDS matrices of order $3 $ over the finite field $\mathbb{F}_{2^m}$. Since our objective is to extend this result to a bigger class of matrices, we now provide the definition of such matrices.

The definition of semi-involutory matrix is the following \cite{CCK}.
\begin{definition}
A non-singular matrix $M$ is said to be semi-involutory if there exist non-singular diagonal matrices $D_1$ and $D_2$ such that $M^{-1} = D_1MD_2$.  
\end{definition}

In the following example we provide a $2 \times 2$ semi-involutory matrix over the finite field $\mathbb{F}_{2^2}$.
\begin{example}
Consider the finite field $\mathbb{F}_{2^2}$ with elements $\{0,1,\alpha,\alpha+1\}$. Let $A=\begin{bmatrix}
1 & \alpha\\
\alpha+1 & \alpha
\end{bmatrix}$. Consider $D_1=D_2=\begin{bmatrix}
\alpha & 0\\
0 & 1\\
\end{bmatrix}$. Then $A^{-1}=D_1AD_2=\begin{bmatrix}
\alpha+1 & \alpha+1\\
1 & \alpha
\end{bmatrix}$.
\end{example}
It is also possible to construct semi-involutory  matrices over fields of odd characteristic.  For example,
\begin{example}
Consider the finite field $\mathbb{F}_{11}$ and $A=\begin{bmatrix}
7 & 3\\
4 & 2\\
\end{bmatrix}$. Here $\det(A)=2$. Consider $D_1=\begin{bmatrix}
4 & 0\\
0 & 8\\
\end{bmatrix}$ and $D_2=\begin{bmatrix}
2 & 0\\
0 & 4\\
\end{bmatrix}.$ Then $A^{-1}=D_1AD_2$.
\end{example}

Some equivalent conditions of semi-involutory matrices are given in \cite{CCK}, which are as follows:
\begin{enumerate}
    \item $A$ is semi-involutory.
    \item $A^{-1}$ and $A^T$ are semi-involutory.
    \item $DAD'$ is semi-involutory for any non-singular diagonal matrices $D$ and $D'$.
    \item $P^TAP$ is semi-involutory for any permutation matrix $P$.
    \item $ADA$ is non-singular and diagonal for some diagonal matrix $D$.
\end{enumerate}
We say a matrix $D$ in the equivalent definition $5$ as an associated diagonal matrix for the semi-involutory matrix $A$.
A matrix $M_1$ is said to be “permutation similar” to another matrix $M_2$ if and only if $M_2$ can be obtained by permutation of any rows (or columns) of $M_1$.
 
\begin{definition}\label{def1}
A square matrix $A$ is said to be
reducible if there exists a permutation matrix $P$ such that $PAP^T=\begin{bmatrix}
A_1 & A_2\\
0 & A_3
\end{bmatrix},$ where $A_1$ and $A_3$ are square matrices of order at least $1$. A matrix is said to be irreducible if
it is not reducible.
\end{definition}

Cheon {\it{et al.}} provided the following characterization of $3 \times 3$ semi-involutory matrices \cite{CCK}. 
%The statement of their result is following.
\begin{theorem}\label{3x3 case}
Let $A=(a_{ij})$ be a real matrix of order $3 \times 3$. Then $A$ is semi-involutory if and only if $A$ is non-singular and one of the following holds.
\begin{itemize}
\item Up to permutation similarity $A$ is a reducible matrix of the form $A=\begin{bmatrix}
B & \mathbf{x}\\
\mathbf{0}^T & c\\
\end{bmatrix}$ such that $B^{-1}=D_1BD_2$ for some non-singular diagonal matrices $D_1$ and $D_2$, and $\mathbf{x}=\mathbf{0}$ or $\mathbf{x}$ is an eigenvector of $BD_1$.
\item Up to permutation similarity $a_{11}=0$ is the only zero entry in $A$, $\det A(1|1)=0$ and $a_{12}a_{23}a_{31} = a_{13}a_{21}a_{32}$. Here, $A(1|1)$ denotes the submatrix of $A$ obtained by removing the first row and first column.
\item A is nowhere zero, i.e., $a_{ij} \neq 0, 1 \leq i,j \leq 3,$ $a_{12}a_{23}a_{31} = a_{13}a_{21}a_{32}$ and $\det X =0$, where $X=\begin{bmatrix}
a_{11}a_{21} & a_{21}a_{22} & a_{23}a_{31}\\
a_{11}a_{31} & a_{21}a_{32} & a_{31}a_{33}\\
a_{12}a_{31} & a_{22}a_{32} & a_{32}a_{33}\\
\end{bmatrix}$.
\end{itemize}
\end{theorem}
Note that Theorem \ref{3x3 case} also holds when the entries of the matrix are from finite fields.

%A square matrix $A$ is said to be reducible if it is permutation-similar to an upper triangular matrix. The matrix $A$ is called irreducible if it is not reducible. 

Cheon {\it{et al.}} also established a relationship between the non-singular diagonal matrices of a semi-involutory matrix using the irreducible property. This result is presented in the following theorem.

\begin{theorem}\label{diag mat cond}
Let $A$ be an irreducible semi-involutory matrix of order $n \times n$ such that $A^{-1}=D_1AD_2$, where $D_1$ and $D_2$ are non-singular diagonal matrices. Then $D_1=cD_2$ for some non-zero constant $c$.
\end{theorem}

\section{Structure of 3$\times$3 semi-involutory MDS matrices}\label{general structures}
In \cite{GSARC}, G\"{u}zel {\it{et al.}} provided a structure to construct all $3 \times 3$ involutory matrices using four arbitrary non-zero elements of the finite field $\mathbb{F}_{2^m}$. We first generalize their result for the case of semi-involutory irreducible matrices in the following theorem. 

\begin{theorem}\label{gen 3x3 inv}
Let $A=\begin{bmatrix}
a_{11} & a_{12} & a_{13}\\
a_{21} & a_{22} & a_{23}\\
a_{31} & a_{32} & a_{33}\\
\end{bmatrix}$ be a $3 \times 3$ irreducible, semi-involutory matrix with an associated diagonal matrix $D=$diagonal$(d_1,d_2,d_3)$ over the finite field $\mathbb{F}_{2^m}$, where $m$ is a positive integer. Then the non-diagonal entries of $A$ can be expressed 
%in terms of the diagonal entries of $A$ and entries of $D$
%, i.e., $a_{11}, a_{22}, a_{33}, d_1, d_2$ and $d_3$ 
as follows:
$a_{12}=(a_{11}d_1+a_{33}d_3)d_2^{-1}x, ~a_{13}= (a_{11}d_1+a_{22}d_2)d_3^{-1}xy,~ a_{21}= (a_{22}d_2+a_{33}d_3)d_1^{-1}x^{-1},~ a_{23}=(a_{22}d_2+a_{11}d_1)d_3^{-1}y,~ a_{31}= (a_{33}d_3+a_{22}d_2)d_1^{-1}(xy)^{-1}$ and $a_{32}= (a_{33}d_3+a_{11}d_1)d_2^{-1}y^{-1},$ 
%and $a_{33}=d_3^{-1}(b+a_{11}d_1+a_{22}d_2)$ 
where $x,y$ are non-zero elements of $\mathbb{F}_{2^m}$.
\end{theorem}

\begin{proof}
Since $A$ is an irreducible, semi-involutory matrix, by Theorem \ref{diag mat cond}, we have $A^{-1}=cDAD$. This implies that $(DA)^2=c^{-1}I$. Let $c^{-1}=a \in \mathbb{F}_{2^m}$. The diagonal and non diagonal entries of $(DA)^2$ satisfy the following conditions:
\begin{eqnarray}
a_{11}^2d_1^2+a_{12}a_{21}d_1d_2+a_{13}a_{31}d_1d_3=a\label{eq 1}\\
a_{22}^2d_2^2+a_{12}a_{21}d_1d_2+a_{23}a_{32}d_2d_3=a\label{eq 2}\\
a_{33}^2d_3^2+a_{13}a_{31}d_1d_3+a_{23}a_{32}d_2d_3=a\label{eq 3}\\
a_{11}a_{12}d_1^2+a_{12}a_{22}d_1d_2+a_{13}a_{32}d_1d_3=0\label{eq 4}\\
a_{11}a_{13}d_1^2+a_{12}a_{23}d_1d_2+a_{13}a_{33}d_1d_3=0\label{eq 5}\\
a_{11}a_{21}d_1d_2+a_{21}a_{22}d_2^2+a_{23}a_{31}d_2d_3=0\label{eq 6}\\
a_{13}a_{21}d_1d_2+a_{22}a_{23}d_2^2+a_{23}a_{33}d_2d_3=0\label{eq 7}\\
a_{11}a_{31}d_1d_3+a_{21}a_{32}d_2d_3+a_{31}a_{33}d_3^2=0\label{eq 8}\\
a_{12}a_{31}d_1d_3+a_{22}a_{32}d_2d_3+a_{32}a_{33}d_3^2=0 \label{eq 9}
\end{eqnarray}

Adding (\ref{eq 1}), (\ref{eq 2}) and (\ref{eq 3}) we get 
\begin{eqnarray}
a_{11}^2d_1^2+a_{22}^2d_2^2+a_{33}^2d_3^2=a \label{eq 10}. 
\end{eqnarray}

Following the idea from Theorem $1$ in \cite{GSARC}, equations (\ref{eq 4}),(\ref{eq 5}),(\ref{eq 6}), (\ref{eq 7}), (\ref{eq 8}) and (\ref{eq 9}) can be re-written  as follows:

\begin{eqnarray}
a_{12}d_1(a_{11}d_1+a_{22}d_2)=a_{13}a_{32}d_1d_3\label{eq 11}\\
a_{13}d_1(a_{11}d_1+a_{33}d_3)=a_{12}a_{23}d_1d_2\label{eq 12}\\
a_{21}d_2(a_{11}d_1+a_{22}d_2)=a_{23}a_{31}d_2d_3\label{eq 13}\\
a_{23}d_2(a_{22}d_2+a_{33}d_3)=a_{13}a_{21}d_1d_2\label{eq 14}\\
a_{31}d_3(a_{11}d_1+a_{33}d_3)=a_{21}a_{32}d_2d_3\label{eq 15}\\
a_{32}d_3(a_{22}d_2+a_{33}d_3)=a_{21}a_{31}d_1d_3\label{eq 16}
\end{eqnarray}

Multiplying equations (\ref{eq 11}) and (\ref{eq 13}), equations (\ref{eq 12}) and (\ref{eq 15}), and equations (\ref{eq 14}) and (\ref{eq 16}), we obtain the following:
\begin{eqnarray}
a_{12}a_{21}d_1d_2(a_{11}d_1+a_{22}d_2)^2=a_{13}a_{31}a_{23}a_{32}d_1d_2d_3^2\label{eq 17}\\
a_{13}a_{31}d_1d_3(a_{11}d_1+a_{33}d_3)^2=a_{12}a_{21}a_{23}a_{32}d_1d_2^2d_3\label{eq 18}\\
a_{23}a_{32}d_2d_3(a_{22}d_2+a_{33}d_3)^2=a_{13}a_{31}a_{12}a_{21}d_1^2d_2d_3\label{eq 19}
\end{eqnarray}

Multiplying equations (\ref{eq 18}), (\ref{eq 19}) we get:
\begin{eqnarray}
(a_{11}d_1+a_{33}d_3)(a_{22}d_2+a_{33}d_3)=a_{12}a_{21}d_1d_2 \label{eq 20}
\end{eqnarray}

Similarly, multiplying equations (\ref{eq 17}) and (\ref{eq 18}), we have:
\begin{eqnarray}
(a_{11}d_1+a_{22}d_2)(a_{11}d_1+a_{33}d_3)=a_{23}a_{32}d_2d_3 \label{eq 21}
\end{eqnarray}
Finally, multiplying equations (\ref{eq 17}) and (\ref{eq 19}), we obtain:
\begin{eqnarray}
(a_{11}d_1+a_{22}d_2)(a_{22}d_2+a_{33}d_3)=a_{13}a_{31}d_1d_3. \label{eq 22}
\end{eqnarray}

%Equation (\ref{eq 4}),(\ref{eq 5}),(\ref{eq 6}), (\ref{eq 7}), (\ref{eq 8}) and (\ref{eq 9}) can be re-written in the following form thanks to an idea follows from Theorem $1$ in \cite{GSARC}:
%\begin{eqnarray}
%(a_{11}d_1+a_{33}d_3)(a_{22}d_2+a_{33}d_3)=a_{12}a_{21}d_1d_2 \label{eq 11}\\
%(a_{11}d_1+a_{22}d_2)(a_{11}d_1+a_{33}d_3)=a_{23}a_{32}d_3d_2 \label{eq 12}\\
%(a_{11}d_1+a_{22}d_2)(a_{22}d_2+a_{33}d_3)=a_{13}a_{31}d_3d_1. \label{eq 13}
%\end{eqnarray}

Multiply the first term in the product by $d_2^{-1}$ and second term by $d_1^{-1}$ in the left hand side of equation (\ref{eq 20}), we can write $a_{12}=(a_{11}d_1+a_{33}d_3)d_2^{-1}x$ and $a_{21}=(a_{22}d_2+a_{33}d_3)d_1^{-1}x^{-1}$ where $x$ is an non-zero element of $\mathbb{F}_{2^m}.$ 
Similarly, form (\ref{eq 21}), we get $a_{23}=(a_{22}d_2+a_{11}d_1)d_3^{-1}y$ and $a_{32}=(a_{33}d_3+a_{11}d_1)d_2^{-1}y^{-1}$, where $y$ is an non-zero element of $\mathbb{F}_{2^m}$. Finally, from (\ref{eq 22}) we get $a_{13}=(a_{11}d_1+a_{22}d_2)d_3^{-1}z$ and $a_{31}=(a_{33}d_3+a_{22}d_2)d_1^{-1}z^{-1}$ where $z$ is an non-zero element of $\mathbb{F}_{2^m}$.
Since $A$ is semi-involutory and irreducible, entries of $A$ satisfy $a_{12}a_{23}a_{31}=a_{13}a_{21}a_{32}$ from Theorem \ref{3x3 case}. This implies $x,y$ and $z$ satisfy $xyz^{-1}=x^{-1}y^{-1}z$. By choosing $z=xy$ we get desired $a_{13}$ and $a_{31}.$
\end{proof}

\begin{remark} \label{gen 3x3 inv example}
Irreducible semi-involutory matrices may not be involutory. Thus Theorem \ref{gen 3x3 inv}
holds for a more general class of matrices. For example, 
consider the finite field $\mathbb{F}_{2^3}$ with generating polynomial $x^3+x^2+1$. Let $\alpha$ be a primitive element of the finite field.
Consider the matrix $A= \begin{bmatrix}
\alpha^2+\alpha & 1 & \alpha^2+1\\
1 & \alpha^2+\alpha & \alpha+1\\
\alpha^2+1 &  \alpha+1 & \alpha^2+\alpha
\end{bmatrix}$. Then $A$ is semi-involutory and irreducible with $D=$ diagonal $(\alpha^2+\alpha+1, \alpha^2+\alpha, \alpha+1)$ and $c=1$  but $A^2 \neq I$. 
\end{remark}
The converse of Theorem \ref{gen 3x3 inv} is not necessarily true. For example,
\begin{example}
Consider the finite field $\mathbb{F}_{2^2}$ with generating polynomial $x^2+x+1$. Let $\beta$ be a primitive element and $a_{11}=1, a_{22}=\beta, a_{33}=\beta+1, d_1=\beta, d_2=\beta+1$ and $d_3=1$. Take $x=\beta, y=\beta+1$. Then the matrix 
\begin{eqnarray}
A=\begin{bmatrix}
a_{11} & (a_{11}d_1+a_{33}d_3)d_2^{-1}x& (a_{11}d_1+a_{22}d_2)d_3^{-1}xy\\
(a_{22}d_2+a_{33}d_3)d_1^{-1}x^{-1} & a_{22} & (a_{11}d_1+a_{22}d_2)d_3^{-1}y\\
(a_{22}d_2+a_{33}d_3)d_1^{-1}(xy)^{-1} & (a_{11}d_1+a_{33}d_3)d_2^{-1}y^{-1} & a_{33}\\
\end{bmatrix} \label{matrix1}
\end{eqnarray}
is equal to $\begin{bmatrix}
1 & \beta+1 & \beta+1\\
\beta+1 & \beta & \beta\\
1 & \beta+1 & \beta+1\\
\end{bmatrix}$ with $\det A$ is zero. Hence $A$ is not semi-involutory.
\end{example}

Observe that equation (\ref{matrix1}) is the generalized form of the matrix provided in Theorem $1$ of \cite{GSARC}. Substituting $d_1=d_2=d_3=1$ and $a=1$ in equation (\ref{eq 10}), we can deduce to the generalized form of $3 \times 3$ involutory matrix.

Under certain condition we prove the converse of  Theorem \ref{gen 3x3 inv} in the following theorem.

\begin{theorem}\label{converse gen strc SI}
Let $A$ be a $3 \times 3$ matrix over $\mathbb{F}_{2^m}$ as described in equation (\ref{matrix1}), where$a_{11},a_{22},a_{33},d_1,d_2,d_3,x,y$ are non-zero and $a_{11}d_1+a_{22}d_2+a_{33}d_3 \neq 0$. Then $A$ is semi-involutory. Moreover, if $a_{11}d_1+a_{22}d_2, a_{22}d_2+a_{33}d_3$ and $ a_{11}d_1+a_{33}d_3$ are non-zero, then $A$ is irreducible. 
\end{theorem}
\begin{proof}
Consider the diagonal matrix $D$ with $d_1,d_2,d_3$ as consecutive diagonal entries. Then $D$ is a non-singular diagonal matrix and $ADA=$ diagonal $(d_1^{-1}a,d_2^{-1}a,d_3^{-1}a)$ where $a$ satisfies equation (\ref{eq 10}). Therefore $ADA$ is a non-singular diagonal matrix. Since $\det A= (a_{11}d_1+a_{22}d_2+a_{33}d_3)^3(d_1d_2d_3)^{-1},$ by the given condition $\det A \neq 0$. Hence from the equivalent condition $5$ of semi-involutory matrix, $A$ is semi-involutory. Observe that if $a_{11}d_1+a_{22}d_2, a_{22}d_2+a_{33}d_3$ and $ a_{11}d_1+a_{33}d_3$ are non-zero, then $A$ cannot be permutation similar to matrix of the form described in Definition \ref{def1}, as all entries would be non-zero. Consequently, $A$ is irreducible.
\end{proof}

Note that, MDS matrices are irreducible. Using this property we prove the following result.

\begin{theorem} \label{SI MDS}
Let $A$ be a $3 \times 3$ matrix over $\mathbb{F}_{2^m}$ as described in equation (\ref{matrix1}), where $a_{11},a_{22},a_{33},d_1,d_2,d_3,x,y$ are non-zero. Then $A$ is a semi-involutory MDS matrix if and only if $a_{11}d_1+a_{22} d_2, a_{11}d_1+a_{33}d_3, a_{22}d_2+a_{33}d_3$ and $ a_{11}d_1+a_{22} d_2+a_{33}d_3$ are non-zero elements of the finite field.
\end{theorem} 

\begin{proof}
Let $A$ be a semi-involutory MDS matrix. 
%Thus $A$ can be written in the form given by Remark \ref{gen struc SI}. 
Then $A^{-1}$ exists and $\det A=(a_{11}d_1+a_{22}d_2+a_{33}d_3)^3(d_1d_2d_3)^{-1}$. This implies $a_{11}d_1+a_{22}d_2+a_{33}d_3$ is non-zero. Observe that, over the finite field $\mathbb{F}_{2^m}$, equation (\ref{eq 10}) can be written as $(a_{11}d_1+a_{22}d_2+a_{33}d_3)^2=a$ and there exists a non-zero element $b$ such that $b^2=a$.
Hence $a_{11}d_1+a_{22}d_2+a_{33}d_3=b$. 
The determinant of all $2 \times 2$ sub-matrices of $A$ are following:\\

\noindent$\begin{vmatrix}
a_{11} & (a_{11}d_1+a_{33}d_3)d_2^{-1}x\\
(a_{22}d_2+a_{33}d_3)d_1^{-1}x^{-1} & a_{22}\\\
\end{vmatrix}=a_{33}bd_3d_1^{-1}d_2^{-1} ,\\
\begin{vmatrix}
a_{11} & (a_{11}d_1+a_{22}d_2)d_3^{-1}xy\\
(a_{22}d_2+a_{33}d_3)d_1^{-1}x^{-1} & (a_{11}d_1+a_{22}d_2)d_3^{-1}y\\
\end{vmatrix}=(a_{11}d_1+a_{22}d_2)byd_3^{-1}d_1^{-1},\\
\begin{vmatrix}
(a_{11}d_1+a_{33}d_3)d_2^{-1}x& (a_{11}d_1+a_{22}d_2)d_3^{-1}xy\\
a_{22} & (a_{11}d_1+a_{22}d_2)d_3^{-1}y\\
\end{vmatrix}=(a_{11}d_1+a_{22}d_2)bxyd_2^{-1}d_3^{-1},\\
\begin{vmatrix}
a_{11} & (a_{11}d_1+a_{33}d_3)d_2^{-1}x\\
(a_{22}d_2+a_{33}d_3)d_1^{-1}(xy)^{-1} & (a_{11}d_1+a_{33}d_3)d_2^{-1}y^{-1} \\
\end{vmatrix}=(a_{11}d_1+a_{33}d_3)by^{-1}d_1^{-1}d_2^{-1},\\
\begin{vmatrix}
a_{11} & (a_{11}d_1+a_{22}d_2)d_3^{-1}xy\\
(a_{22}d_2+a_{33}d_3)d_1^{-1}(xy)^{-1}& a_{33}\\
\end{vmatrix}=a_{22}bd_2d_1^{-1}d_3^{-1},\\
\begin{vmatrix}
(a_{11}d_1+a_{33}d_3)d_2^{-1}x& (a_{11}d_1+a_{22}d_2)d_3^{-1}xy\\
(a_{11}d_1+a_{33}d_3)d_2^{-1}y^{-1} & a_{33}\\
\end{vmatrix}=(a_{11}d_1+a_{33}d_3)bxd_2^{-1}d_3^{-1},\\
\begin{vmatrix}
(a_{22}d_2+a_{33}d_3)d_1^{-1}x^{-1} & a_{22} \\
(a_{22}d_2+a_{33}d_3)d_1^{-1}(xy)^{-1} & (a_{11}d_1+a_{33}d_3)d_2^{-1}y^{-1} &\\
\end{vmatrix}=(a_{22}d_2+a_{33}d_3)bx^{-1}y^{-1}d_1^{-1}d_2^{-1}, \\
\begin{vmatrix}
(a_{22}d_2+a_{33}d_3)d_1^{-1}x^{-1} & (a_{11}d_1+a_{22}d_2)d_3^{-1}y\\
(a_{22}d_2+a_{33}d_3)d_1^{-1}(xy)^{-1} & a_{33}\\
\end{vmatrix}=(a_{22}d_2+a_{33}d_3)bxd_1^{-1}d_3^{-1} ,\\
\begin{vmatrix}
a_{22} & (a_{11}d_1+a_{22}d_2)d_3^{-1}y\\ (a_{11}d_1+a_{33}d_3)d_2^{-1}y^{-1} & a_{33}\\\end{vmatrix}=a_{11}bd_1d_2^{-1}d_3^{-1}$.

Since $A$ is an MDS matrix and $d_1, d_2, d_3, x, y, b$ are non-zero elements of the finite field, then the condition holds.

Conversely, let  $a_{11}d_1+a_{22} d_2, a_{11}d_1+a_{33}d_3, a_{22}d_2+a_{33}d_3$ and $ a_{11}d_1+a_{22}d_2+a_{33}d_3$ are non-zero elements of the finite field. Then we have all the entries of the matrix $A$ and all $2 \times 2$ sub-matrices have non-zero determinant. Since $ a_{11}d_1+a_{22}d_2+a_{33}d_3 \neq 0,$ 
$\det A$ is also non-zero. Therefore  $A$ is an MDS matrix.
Let $ADA=(M_{ij}), 1 \leq i,j \leq 3$, where $D$= diagonal$(d_1,d_2,d_3)$. Then 
\begin{align*}
M_{11}&=a_{11}^2d_1+(a_{11}d_1+a_{33}d_3)(a_{22}d_2+a_{33}d_3)d_1^{-1}+(a_{11}d_1+a_{22}d_2)(a_{22}d_2+a_{33}d_3)d_1^{-1}\\
&=d_1^{-1}(a_{11}d_1+a_{22}d_2+a_{33}d_3)^2.
\end{align*}
\begin{align*}
M_{22}&=(a_{22}d_2+a_{33}d_3)(a_{22}d_2+a_{33}d_3)d_2^{-1}+a_{22}^2d_2+(a_{11}d_1+a_{22}d_2)(a_{11}d_1+a_{33}d_3)d_2^{-1}\\
&=d_2^{-1}(a_{11}d_1+a_{22}d_2+a_{33}d_3)^2.
\end{align*}
By similar calculations, we have $M_{33}=d_3^{-1}(a_{11}d_1+a_{22}d_2+a_{33}d_3)^2$.
Now, 
\begin{align*}
M_{12}&=a_{11}d_1(a_{11}d_1+a_{33}d_3)d_2^{-1}x+a_{22}(a_{11}d_1+a_{33}d_3)x+(a_{11}d_1+a_{33}d_3)(a_{11}d_1+a_{22}d_2)d_2^{-1}x\\
&=(a_{11}d_1+a_{33}d_3)(a_{11}d_1d_2^{-1}x+a_{22}x+a_{11}d_1d_2^{-1}x+a_{22}x)=0.
\end{align*}
This holds because the characteristic of the finite field is $2$. By similar calculations, all other non-diagonal entries of $ADA$ are also zero. Therefore $ADA=D'$ where $D'=$ diagonal$(d_1^{-1}a,d_2^{-1}a,d_3^{-1}a)$, with $a=(a_{11}d_1+a_{22}d_2+a_{33}d_3)^2$. Since $a$ is a non-zero element of the finite field and $D, D'$ are non-singular matrices, $A$ is semi-involutory. 
\end{proof}
An example illustrating this result is following: 
\begin{example}
Consider the finite field $\mathbb{F}_{2^4}$ with generating polynomial $x^4+x^3+1$. Let $\alpha$ be a primitive element of this field. Choose $a_{11}=1, a_{22}=\alpha, a_{33}=\alpha^2, d_1=\alpha, d_2=\alpha, d_3=\alpha^3+1$. Also set $x=1$ and $y=\alpha$. Then the expressions $a_{11}d_1+a_{22} d_2, a_{11}d_1+a_{33}d_3, a_{22}d_2+a_{33}d_3$ and $ a_{11}d_1+a_{22} d_2+a_{33}d_3$ are non-zero elements. 

Construct the matrix in the form described in Equation (\ref{matrix1}). This yields $$A=\begin{bmatrix}
1 & \alpha^3+\alpha & \alpha^3+\alpha\\
\alpha^3+1 & \alpha & \alpha^3+\alpha\\
\alpha^3 & \alpha^2+1 & \alpha^2
\end{bmatrix}.$$ Also observe that 
 $ADA=D',$ where $D$ is the diagonal matrix with entries $d_1,d_2,d_2$ respectively, and $D'$ is the diagonal matrix with diagonal entries $d_1^{-1}(a_{11}d_1+a_{22}d_2+a_{33}d_3)^2=\alpha^2+\alpha+1, d_2^{-1}(a_{11}d_1+a_{22}d_2+a_{33}d_3)^2=\alpha^2+\alpha+1, d_3^{-1}(a_{11}d_1+a_{22}d_2+a_{33}d_3)^2=\alpha^3+1$, respectively. The matrix $A$ is also an MDS matrix.

This matrix $A$ is irreducible and semi-involutory matrix. It can also be verified that the non-diagonal entries match the form described in Theorem \ref{gen 3x3 inv} with $x=1$ and $y=\alpha$.
\end{example}

\section{A counting problem}\label{cardinality}

In this section, we count the number of semi-involutory MDS matrices over the finite field $\mathbb{F}_{2^m}$. We start with the following construction of a set of $6$-tuples that satisfy the conditions presented in Theorem \ref{SI MDS} over the finite field $\mathbb{F}_{2^m}$: 
$S=\{(a_{11}, a_{22}, a_{33}, d_1, d_2, d_3) \in (\mathbb{F}_{2^m}^*)^6, m \geq 2:  a_{11}d_1 + a_{22}d_2 \neq 0, a_{11}d_1 +a_{33}d_3 \neq 0, a_{22}d_2 + a_{33}d_3 \neq 0, a_{11}d_1 + a_{22}d_2 + a_{33}d_3 \neq 0\}$. Using the cardinality of $S$, we count the number of semi-involutory MDS matrices in Theorem \ref{counting result}. To determine the cardinality of $S$, we partition the set $S$ into five disjoint non-empty subsets. Note that the cardinality of these subsets are comparatively easier to count which are described in the following lemmas. Finally, we add all the cardinalities to get Theorem \ref{counting result}. Here are the counting lemmas: 

\begin{lemma}\label{1st set}
Let $S_1= \{(a_{11}, a_{22}, a_{33}, d_1, d_2, d_3) \in (\mathbb{F}_{2^m}^*)^6, m \geq 2: a_{ii} \neq a_{jj}, 1 \leq i<j \leq 3, a_{11}d_1 + a_{22}d_2 \neq 0, a_{11}d_1 +a_{33}d_3 \neq 0, a_{22}d_2 + a_{33}d_3 \neq 0, a_{11}d_1 + a_{22}d_2 + a_{33}d_3 \neq 0\}$. Then $|S_1|= (2^m-1)^2(2^m-2)^2(2^m-3)(2^m-4)$.
\end{lemma}
\begin{proof}
Let $a_{ii} \neq a_{jj}, 1 \leq i<j \leq 3$. We consider three sub-cases based on the choice of $d_1, d_2, d_3$. As a result, let consider the following three sets:
\begin{align*}
S_1'&=\{(a_{11}, a_{22}, a_{33}, d_1, d_2, d_3) \in (\mathbb{F}_{2^m}^*)^6, m \geq 2: a_{ii} \neq a_{jj}, d_i=d_j, 1 \leq i<j \leq 3 ,\\
& ~~~a_{11}d_1 + a_{22}d_2 \neq 0, a_{11}d_1 +a_{33}d_3 \neq 0, a_{22}d_2 + a_{33}d_3 \neq 0, a_{11}d_1 + a_{22}d_2 + a_{33}d_3 \neq 0\}\\
S_1''&=\{(a_{11}, a_{22}, a_{33}, d_1, d_2, d_3) \in (\mathbb{F}_{2^m}^*)^6, m \geq 2: a_{ii} \neq a_{jj},1 \leq i<j \leq 3,  d_i=d_j ~\text{for}~\\
&~~~(i,j)\in \{(1,2), (1,3),(2,3)\}, a_{11}d_1 + a_{22}d_2 \neq 0, a_{11}d_1 +a_{33}d_3 \neq 0, a_{22}d_2 + a_{33}d_3 \neq 0,\\
&~~~ a_{11}d_1 + a_{22}d_2 + a_{33}d_3 \neq 0\}\\
S_1'''&=\{(a_{11}, a_{22}, a_{33}, d_1, d_2, d_3) \in (\mathbb{F}_{2^m}^*)^6, m \geq 2: a_{ii} \neq a_{jj}, d_i \neq d_j, 1 \leq i<j \leq 3 ,\\
& ~~~a_{11}d_1 + a_{22}d_2 \neq 0, a_{11}d_1 +a_{33}d_3 \neq 0, a_{22}d_2 + a_{33}d_3 \neq 0, a_{11}d_1 + a_{22}d_2 + a_{33}d_3 \neq 0\}
\end{align*} 

\textbf{Case I.} First consider the set $S_1'$. In this case we have $d_i=d_j$ for $1 \leq i<j \leq 3 $.
This implies that $a_{11}d_1+a_{22}d_2, a_{11}d_1+a_{33}d_3, a_{22}d_2+a_{33}d_3$ are non-zero. Since $d_i \in \mathbb{F}_{2^m}^* $, the equation $a_{11}d_1+a_{22} d_2+a_{33}d_3=0$ holds if $a_{11}+a_{22}+a_{33}=0$. 
%The cardinality of the set $A_1=\{(a_{11},a_{22},a_{33}):a_{11}+a_{22}+a_{33}=0\}$ is $(2^m-1)(2^m-2)$. 
Consider the sets $A_1=\{(a_{11},a_{22},a_{33}):a_{11}+a_{22}+a_{33}=0\}$ and $A_2=\{(a_{11},a_{22},a_{33}):a_{11}+a_{22}+a_{33} \neq 0\}$. Any non-zero $d_i \in \mathbb{F}_{2^m}$ satisfy all the four conditions of $S_1'$ for each element of $A_2$. Clearly,
\begin{align*}
|A_2| = (2^m-1)(2^m-2)(2^m-3)-(2^m-1)(2^m-2)=(2^m-1)(2^m-2)(2^m-4).
\end{align*}
Hence cardinality of $S_1'$ is \begin{align*}
|S_1'|&=\{(2^m-1)(2^m-2)(2^m-3)-(2^m-1)(2^m-2)\}(2^m-1)\\
&=(2^m-1)^2(2^m-2)(2^m-4).
\end{align*} 
%$\{(2^m-1)(2^m-2)(2^m-3)-(2^m-1)(2^m-2)\}(2^m-1)=(2^m-1)^2(2^m-2)(2^m-4)$. 

\textbf{Case II.} Next consider the set $S_1''$. Then exactly one pair of $d_i$ is equal.
We give the proof for the case $d_1= d_2,  d_1 \neq d_3, d_2 \neq d_3$. The other two cases will follow similarly.

Let $d_1= d_2,  d_1 \neq d_3, d_2 \neq d_3$. Then $a_{11}d_1+a_{22}d_2$ is non-zero. To determine the cardinality of $S_1''$, we count the cardinality of non-zero $3$-tuple $(d_1,d_1,d_3)$ such that $ a_{11}d_1+a_{33}d_3, a_{22}d_1+a_{33}d_3, a_{11}d_1+a_{22}d_1+a_{33}d_3 $ are non-zero.

First we choose an arbitrary  $(a_{11},a_{22},a_{33})$ such that $a_{33} \neq a_{11}+a_{22}$ and fix it. There are $(2^m-1)$ and $(2^m-2)$ ways to choose $a_{11}$ and $a_{22}$ respectively. Since $a_{33}$ is different from $a_{11}+a_{22}$, we have $(2^m-4)$ many options for $a_{33}$. Therefore $(a_{11},a_{22},a_{33})$ can be chosen in total $(2^m-1)(2^m-2)(2^m-4)$ ways. Let define the following two sets:
\begin{align*}
T &=\{(d_1,d_3) \neq (0,0): d_1 \neq d_3\}\\
\text{and}~ X &=\{(d_1,d_3)\in T:a_{11}d_1+a_{33}d_3 \neq 0,a_{22}d_1+a_{33}d_3 \neq 0, a_{11}d_1+a_{22}d_1+a_{33}d_3 \neq 0 \}.
\end{align*} 
Clearly, $|T|=(2^m-1)(2^m-2)$.  
To count the cardinality of $X$, consider three sets 
\begin{align*}
X_1&=\{(d_1,d_3) \in T : a_{11}d_1+a_{33}d_3 = 0\},\\
X_2&=\{(d_1,d_3) \in T : a_{22}d_1+a_{33}d_3 = 0\} \\
\text{and}~~X_3 &=\{(d_1,d_3) \in T :( a_{11}+a_{22})d_1+a_{33}d_3 = 0\}
\end{align*}
and observe that $|X|=|T|\setminus |X_1 \cup X_2 \cup X_3|$. To calculate the cardinality of $X_i, 1 \leq i \leq 3$, we construct three sets $Y_1, Y_2$ and $Y_3$ defined as follows: 
\begin{align*}
Y_1&=\{(xa_{33},xa_{11}): x \in \mathbb{F}_{2^m}^*\},\\
Y_2&=\{(ya_{33},ya_{22}): y \in \mathbb{F}_{2^m}^*\}\\
\text{and}~ Y_3&=\{(za_{33},z(a_{11}+a_{22}): z \in \mathbb{F}_{2^m}^*\}.
\end{align*} 
We prove that $X_i=Y_i, 1 \leq i \leq 3$. First observe that, $Y_1 \subseteq X_1,Y_2\subseteq X_2$ and $Y_3\subseteq X_3$.

To prove the other side of inclusion, i.e., $X_1 \subseteq Y_1$, consider an arbitrary element $(\alpha_1,\alpha_2 ) \in X_1$. Therefore $a_{11}\alpha_1+a_{33}\alpha_2 = 0$. Moreover, there exist non-zero elements $\beta_1, \beta_2$ in $\mathbb{F}_{2^m}^*$ such that $\alpha_1=\beta_1a_{33}$ and $\alpha_2=\beta_2a_{11}$. Hence $a_{11}a_{33}(\beta_1+\beta_2)=0 $ implies $\beta_1=\beta_2$. Thus $(\alpha_1,\alpha_2 ) \in Y_1$ and this implies $X_1=Y_1$. Similarly, $Y_2=X_2 $ and $Y_3=X_3$ and $|X_1|= |X_2|= |X_3|=(2^m-1)$.

Next assume that  $X_1 \cap X_2 \neq \phi$. This implies $a_{11}=a_{22}$, which is not possible. Similarly, if $X_1 \cap X_3 \neq \phi$, that implies $a_{11}=a_{11}+a_{22}$, which is not possible. Also, if $X_2 \cap X_3 \neq \phi$, then $a_{22}=a_{11}+a_{22}$, which is not possible.
Therefore,
\begin{align*}
|X_1\cup X_2\cup X_3| &=3(2^m-1)\\\text{and}
~|X| =|T| \setminus |X_1\cup X_2\cup X_3|&= (2^m-1)(2^m-2)-3(2^m-1)=(2^m-1)(2^m-5).
\end{align*}

For the case $a_{33}= a_{11}+a_{22}$, first we choose arbitrary $a_{11},a_{22}$ and fix them. 
%There are $(2^m-1)(2^m-2)$ ways to choose the triplets $(a_{11},a_{22},a_{33})$. Choose an arbitrary triplets $(a_{11},a_{22},a_{33})$ and fix it. 
Then $a_{11}d_1+a_{22}d_1+a_{33}d_3=(d_1+d_3)a_{33} \neq 0$. Hence, in this situation the set $X$ is defined as $X=\{(d_1,d_3)\in T:a_{11}d_1+a_{33}d_3 \neq 0,a_{22}d_1+a_{33}d_3 \neq 0 \}$. Following a similar approach as in previous case, construct the sets $X_1,X_2,Y_1,Y_2$ and count $|X|$. Then
\begin{align*}
|X|=(2^m-1)(2^m-2)-2(2^m-1)=(2^m-1)(2^m-4).
\end{align*}
Combining both cases for $d_1= d_2, d_1 \neq d_3, d_2 \neq d_3$ and considering the remaining two cases for $d_i$'s, we have
\begin{align*}
|S_1''|&=3\{(2^m-1)^2(2^m-2)(2^m-4)(2^m-5)+(2^m-1)^2(2^m-2)(2^m-4)\}\\&=3(2^m-1)^2(2^m-2)(2^m-4)^2.
\end{align*}

\textbf{Case III.} Lastly, consider $S_1'''$.   

First choose an arbitrary triplet $(a_{11},a_{22},a_{33})$ and fix it. First we define the following two sets: 
% such that $a_{33}=a_{11}+a_{22}$. 
\begin{align*}
T &=\{(d_1,d_2,d_3) \neq (0,0,0): d_i \neq d_j , 1 \leq i<j \leq 3\}\\
X&=\{(d_1,d_2,d_3) \in T: a_{11}d_1 + a_{22}d_2 \neq 0, a_{11}d_1 + a_{33}d_3 \neq 0, a_{22}d_2 + a_{33}d_3 \neq 0,\\
& a_{11}d_1 + a_{22}d_2+ a_{33}d_3 \neq 0\}
\end{align*}
Clearly, $|T|= (2^m-1)(2^m-2)(2^m-3)$.

Consider the sets 
\begin{align*}
X_1 &=\{(d_1,d_2,d_3) \in T: a_{11}d_1 + a_{22}d_2 = 0\},\\
X_2 &=\{(d_1,d_2,d_3) \in T: a_{11}d_1 + a_{33}d_3 = 0\},\\
X_3 &=\{(d_1,d_2,d_3) \in T: a_{22}d_2 + a_{33}d_3 = 0\},\\
\text{and}~~X_4 &=\{(d_1,d_2,d_3) \in T: a_{11}d_1 + a_{22}d_2+a_{33}d_3 = 0\}. 
\end{align*}
First we calculate $|X_1\cup X_2\cup X_3\cup X_4|$. We begin with defining the following sets
\begin{align*}
Y_1 &=\{(xa_{22},xa_{11},d_3): x \in \mathbb{F}_{2^m}^*, d_3 \neq (0, xa_{22}, xa_{11})\},\\
Y_2 &=\{(ya_{33},d_2,ya_{11}): y \in \mathbb{F}_{2^m}^*,d_2 \neq ( 0, ya_{33}, ya_{11})\},\\
Y_3 &=\{(d_1,za_{33},za_{22}): d_1 \in \mathbb{F}_{2^m}^*, z \in \mathbb{F}_{2^m}^*, za_{33}\neq d_1 , za_{11} \neq d_1\}\\
&=\{(d_1,za_{33},za_{22}):d_1 \in \mathbb{F}_{2^m}^*, z \in \mathbb{F}_{2^m}^*,z \neq  (d_1a_{33}^{-1}, d_1a_{22}^{-1} )\},\\
\text{and}~~ Y_4 &=\{(d_1,d_2,d_3) : d_1 \in \mathbb{F}_{2^m}^*, d_2 \neq (0,  d_1, a_{22}^{-1}a_{11}d_1, (a_{11}+a_{33})a_{22}^{-1}d_1, (a_{22}+a_{33})^{-1}a_{11}d_1),\\
& d_3 \neq 0, d_3=a_{33}^{-1}(a_{11}d_1 + a_{22}d_2)\}.
\end{align*} We prove that $X_i=Y_i, 1 \leq i \leq 4$.
For the case $a_{33}=a_{11}+a_{22}$, $Y_4$ is denoted by 
\begin{align*}
Y_4^0 &=\{(d_1,d_2,d_3):d_1 \in \mathbb{F}_{2^m}^*, d_2 \neq (0, d_1, a_{22}^{-1}a_{11}d_1), d_3 \neq 0, d_3=a_{33}^{-1}(a_{11}d_1 + a_{22}d_2)\}.
\end{align*}
The cardinality of $Y_4^0$ and $Y_4$ are  $|Y_4^0|=(2^m-1)(2^m-3)$  
and $|Y_4|=(2^m-1)(2^m-5)$.

Clearly, $Y_1\subseteq X_1, Y_2\subseteq X_2$ and $ Y_3\subseteq X_3$. We now prove the reverse inclusion starting with the case $X_1\subseteq Y_1$, and the other two cases are similar.

Let $(d_1',d_2',d_3') \in X_1.$ From the construction of $X_1$ this implies $a_{11}d_1' + a_{22}d_2' = 0$
and $d_3' \neq \{d_1', d_2'\}$. There exists non-zero elements $\beta_1,\beta_2 \in \mathbb{F}_{2^m}^*$ such that $d_1'=\beta_1a_{22}, d_2'=\beta_2a_{11}$. Since $a_{11}d_1' + a_{22}d_2' = 0$, we have $a_{11}a_{22}(\beta_1+\beta_2)=0$. Since $a_{11},a_{22}$ are non-zero element, this implies $\beta_1=\beta_2$. Therefore $d_3' \neq \{ \beta_1a_{22}, \beta_1a_{11}\}$ and $X_1=Y_1$. Hence, cardinality of the set $Y_1$ is $(2^m-1)(2^m-3)$ because, we need to choose $x$ from $\mathbb{F}_{2^m}^*$ and $d_3$ from $\mathbb{F}_{2^m}^*\setminus \{0, xa_{22}, xa_{11}\}$. 

 Likewise, $X_2=Y_2$ and $ X_3=Y_3$. Hence $|X_1|=|X_2|=|X_3|=(2^m-1)(2^m-3)$.

We now prove that $Y_4=X_4$ for the case $a_{33} \neq a_{11}+a_{22}$.
%for the case $a_{33} \neq a_{11}+a_{22}$.\\
Let $(d_1',d_2',d_3')\in Y_4$. Then $ d_3'=a_{33}^{-1}(a_{11}d_1' + a_{22}d_2')$ and this implies $ a_{11}d_1' + a_{22}d_2'+a_{33}d_3'=0 $. Also $d_3' \neq 0$ implies $a_{11}d_1' \neq a_{22}d_2'$ i.e., $d_2' \neq a_{22}^{-1}a_{11}d_1'$. 

To prove $(d_1',d_2',d_3')\in X_4$, we need to show $d_3' \neq \{d_1',d_2'\}$. If $d_3'= d_1'$ then $d_1' = a_{33}^{-1}(a_{11}d_1' + a_{22}d_2')$, which implies $d_2'=(a_{11}+a_{33})a_{22}^{-1}d_1'$, which is not possible from the construction of $Y_4$. Similarly, if $d_3'=d_2',$ then $d_2' = a_{33}^{-1}(a_{11}d_1' + a_{22}d_2')$, which implies  $d_2'=(a_{22}+a_{33})^{-1}a_{11}d_1'$, which is also not possible from the construction of $Y_4$. Therefore, $(d_1',d_2',d_3') \in T$ and $Y_4 \subseteq X_4$.

Conversely , let $(d_1',d_2',d_3')\in X_4$. Then $d_1' \neq d_2', d_1'\neq d_3', d_1' \neq d_2'$ and $a_{11}d_1' + a_{22}d_2' + a_{33}d_3'=0$. Since $ a_{33}d_3' \neq 0$, we have $a_{11}d_1' \neq a_{22}d_2'$ and this implies $d_2' \neq a_{22}^{-1}a_{11}d_1'$. Furthermore, since $d_3'\neq d_1', d_2'$, it follows that $a_{11}d_1' + a_{22}d_2' \neq a_{33}d_1'$ and $a_{11}d_1' + a_{22}d_2' \neq a_{33}d_2'$. This implies $d_2' \neq (a_{11}+a_{33})a_{22}^{-1}d_1'$ and $d_2' \neq (a_{22}+a_{33})^{-1}a_{11}d_1'$ respectively. Therefore $(d_1',d_2',d_3')\in Y_4.$

Similarly it can be proved that $Y_4^0=X_4$ when $a_{33}=a_{11}+a_{22}$.

Next, we calculate the cardinality of $ X_i\cap X_j, 1 \leq i,j \leq 4$. We start with calculating $X_1\cap X_2$ and $X_1\cap X_4$, and the others cases are similar.

Let $(d_1',d_2',d_3')\in X_1\cap X_2$. This implies $d_1'=xa_{22}=ya_{33}, d_2'=xa_{11}, d_3'=ya_{11}$.
% i.e, $x=a_{11}^{-1}d_2', y=a_{11}^{-1}d_3'$. 
Since $a_{22} \neq a_{33}$, and $xa_{22}=ya_{33}=d_1'$ then $x=d_1'a_{22}^{-1}$ and $y=d_1'a_{33}^{-1}$. Therefore $\{(d_1',d_1'a_{22}^{-1}a_{11},d_1'a_{33}^{-1}a_{11}): d_1' \in \mathbb{F}_{2^m}^*\} = X_1 \cap X_2$ . Then $|X_1 \cap X_2|=(2^m-1)$. Similar to this, $|X_1\cap X_3|= |X_2\cap X_3|=(2^m-1)$.

Let $(d_1',d_2',d_3')\in X_1\cap X_4$. Then $a_{11}d_1' + a_{22}d_2' = 0$ and $a_{11}d_1' + a_{22}d_2'+a_{33}d_3' = 0$. This implies $a_{33}d_3'=0$, which is not possible. Thus $|X_1 \cap X_4|= \phi$. Similarly, $ |X_2 \cap X_4|= |X_3\cap X_4| = \phi$. 

Lastly, we calculate the cardinality of the set $X_1 \cap X_2 \cap X_3$. For that, we prove that $X_1 \cap X_2 \cap X_3=X_1 \cap X_2 $.
 
Let $(d_1',d_2',d_3') \in X_1\cap X_2$. then $a_{11}d_1' + a_{22}d_2' = 0$ and $a_{11}d_1' + a_{33}d_3' =0$. Adding these, we have $a_{22}d_2' +a_{33}d_3'=0$. Thus $(d_1',d_2',d_3') \in X_1\cap X_2\cap X_3$. Other side of the inclusion follow easily. Thus $|X_1 \cap X_2 \cap X_3|=(2^m-1)$.

Therefore $|X_1 \cup X_2 \cup X_3 \cup X_4|=\left.
  \begin{cases}
  2(2^m-1)(2\cdot2^m-7),  & \text{for } a_{11}+a_{22}=a_{33}\\
  (2^m-1)(4\cdot2^m-16) ,  & \text{for } a_{11}+a_{22}\neq a_{33}
   \end{cases}
   \right. $
and

$|X|=\left.
  \begin{cases}
  (2^m-1)(2^{2m}-9.2^m+20),  & \text{for } a_{11}+a_{22}=a_{33}\\
  (2^m-1)(2^{2m}-9.2^m+22),  & \text{for } a_{11}+a_{22}\neq a_{33}.
  \end{cases}
   \right. $

Cardinality of $S_1'''$ in this case is $(2^m-1)^2(2^m-2)\{(2^m-4)(2^{2m}-9.2^m+22)+(2^{2m}-9.2^m+20)\}$. 
Since any two of $S_1',S_1''$ and $S_1'''$ have empty intersection, 
%for the case $a_{ii} \neq a_{jj}, 1 \leq i<j \leq 3$, 
cardinality of $S_1$ is: 
\begin{align*}
|S_1|&=(2^m-1)^2(2^m-2)(2^m-4)+3(2^m-1)^2(2^m-2)(2^m-4)^2+(2^m-1)^2(2^m-2)\\
&~~~\cdot\{(2^m-4)(2^{2m}-9\cdot2^m+22)+(2^{2m}-9\cdot2^m+20)\}\\
&=(2^m-1)^2(2^m-2)\{(2^m-4)+3(2^m-4)^2+(2^m-4)(2^{2m}-9\cdot2^m+22)\\&~~~~~+(2^{2m}-9\cdot2^m+20)\}\\
&=(2^m-1)^2(2^m-2)\{(2^m-4)(2^{2m}-6\cdot2^m+11)+(2^{2m}-9\cdot2^m+20)\}\\
&=(2^m-1)^2(2^m-2)\{(2^m-4)(2^{2m}-6\cdot2^m+11)+(2^m-4)(2^m-5)\}\\
&=(2^m-1)^2(2^m-2)(2^m-4)(2^{2m}-5\cdot2^m+6)\\
&=(2^m-1)^2(2^m-2)^2(2^m-3)(2^m-4).
\end{align*}
\end{proof}

In the next lemma, we consider another condition on the $a_{ii}$'s and determine the cardinality of the set $S_2$ derived from $S$.

\begin{lemma}\label{2nd set}
Let $S_2= \{(a_{11}, a_{22}, a_{33}, d_1, d_2, d_3) \in (\mathbb{F}_{2^m}^*)^6, m \geq 2:~ a_{ii} = a_{jj}, 1 \leq i<j \leq 3, a_{11}d_1 + a_{22}d_2 \neq 0, a_{11}d_1 +a_{33}d_3 \neq 0, a_{22}d_2 + a_{33}d_3 \neq 0, a_{11}d_1 + a_{22}d_2 + a_{33}d_3 \neq 0\}$. Then $|S_2|=(2^m-1)^2(2^m-2)(2^m-4)$.
\end{lemma}
\begin{proof}
Assume that $a_{ii} = a_{jj}, 1 \leq i<j \leq 3$. In this case, the only possible choice for $d_i$'s are $d_i \neq d_j, 1 \leq i<j \leq 3$.

As a result $a_{11}d_1+a_{22}d_2, a_{11}d_1+a_{33}d_3, a_{22}d_2+a_{33}d_3$ are never zero. Additionally, the elements in the set $S_2$ satisfy the condition $a_{11}d_1+a_{22}d_2+a_{33}d_3=a_{11}(d_1+d_2+d_3) \neq 0$. Since $a_{11}$ in non-zero, we need $d_1+d_2+d_3$ cannot be zero i.e., $d_3 \neq d_1+d_2$. Then for the triplet $(d_1,d_2,d_3)$, we have
\begin{align*}
(2^m-1)(2^m-2)(2^m-3)-(2^m-1)(2^m-2)=(2^m-1)(2^m-2)(2^m-4)
\end{align*} many choices. Therefore cardinality of $S_2$ is $(2^m-1)^2(2^m-2)(2^m-4)$. 
\end{proof}

For the last case, let us assume at most one pair of $a_{ii}$'s are equal i.e., $a_{ii}=a_{jj}$ for $(i,j) \in \{(1,2),(1,3),(2,3)\}$. We prove for the case $a_{11}=a_{22}, a_{11} \neq a_{33}$ and the other two cases follow similarly. 

\begin{lemma}\label{3rd set}
Let $S_3= \{(a_{11}, a_{22}, a_{33}, d_1, d_2, d_3) \in (\mathbb{F}_{2^m}^*)^6, m \geq 2: a_{11}=a_{22}, a_{11} \neq a_{33}, a_{11}d_1 + a_{22}d_2 \neq 0, a_{11}d_1 +a_{33}d_3 \neq 0, a_{22}d_2 + a_{33}d_3 \neq 0, a_{11}d_1 + a_{22}d_2 + a_{33}d_3 \neq 0\}$. Then $|S_3|=(2^m-1)^2(2^m-2)(2^{2m}-6 \cdot2^m+8)$.
\end{lemma}
\begin{proof}
Let $a_{11}=a_{22}, a_{11} \neq a_{33}, a_{22} \neq a_{33}$. We study each sub-cases of $d_i$ separately. Let
\begin{align*}
S_3' &=\{(a_{11}, a_{22}, a_{33}, d_1, d_2, d_3) \in (\mathbb{F}_{2^m}^*)^6, m \geq 2: a_{11}=a_{22}, a_{11} \neq a_{33}, d_1 = d_2, d_2 \neq d_3,\\
&~~ d_1 \neq d_3, a_{11}d_1 + a_{22}d_2 \neq 0, a_{11}d_1 +a_{33}d_3 \neq 0, a_{22}d_2 + a_{33}d_3 \neq 0,\\
&~~ a_{11}d_1 + a_{22}d_2 + a_{33}d_3 \neq 0\},\\
S_3'' &=\{(a_{11}, a_{22}, a_{33}, d_1, d_2, d_3) \in (\mathbb{F}_{2^m}^*)^6, m \geq 2: a_{11}=a_{22}, a_{11} \neq a_{33}, d_1\neq d_2, d_2 =d_3, \\
&~~d_1 \neq d_3, a_{11}d_1 + a_{22}d_2 \neq 0, a_{11}d_1 +a_{33}d_3 \neq 0, a_{22}d_2 + a_{33}d_3 \neq 0,\\
&~~ a_{11}d_1 + a_{22}d_2 + a_{33}d_3 \neq 0\},\\
S_3''' &=\{(a_{11}, a_{22}, a_{33}, d_1, d_2, d_3) \in (\mathbb{F}_{2^m}^*)^6, m \geq 2: a_{11}=a_{22}, a_{11} \neq a_{33},   d_1\neq d_2, d_2 \neq d_3,\\ 
&~~ d_1 = d_3, a_{11}d_1 + a_{22}d_2 \neq 0, a_{11}d_1 +a_{33}d_3 \neq 0, a_{22}d_2 + a_{33}d_3 \neq 0,\\
&~~ a_{11}d_1 + a_{22}d_2 + a_{33}d_3 \neq 0\},\\
S_3'''' &=\{(a_{11}, a_{22}, a_{33}, d_1, d_2, d_3) \in (\mathbb{F}_{2^m}^*)^6, m \geq 2: a_{11}=a_{22}, a_{11} \neq a_{33},  d_1 \neq d_2, d_2 \neq d_3,\\
&~~ d_1 \neq d_3, a_{11}d_1 + a_{22}d_2 \neq 0, a_{11}d_1 +a_{33}d_3 \neq 0, a_{22}d_2 + a_{33}d_3 \neq 0,\\
&~~ a_{11}d_1 + a_{22}d_2 + a_{33}d_3 \neq 0\}.
\end{align*}

\textbf{Case I.} Consider the set $S_3'$. In this set, we have $d_1 = d_2, d_2 \neq d_3, d_1 \neq d_3$. Therefore $a_{11}d_1+a_{22}d_2$ is always zero and this case will never occur.

\textbf{Case II.} Consider the set $S_3''$. Then $d_1 \neq d_2, d_2=d_3, d_1 \neq d_3$.

For each value of $d_i$'s with $1 \leq i \leq 3$, both $a_{22}d_2+a_{33}d_3$ and $a_{11}d_1+a_{22}d_2$ are always non-zero in this case.

Furthermore, the elements of the set $S_3''$ satisfy $ a_{11}d_1+a_{33}d_3, a_{11}d_1+a_{22}d_2+a_{33}d_3 $ are non-zero.
To find the cardinality of $S_3''$, we first fix an arbitrary $3$-tuple $(a_{11},a_{22},a_{33}) \in (\mathbb{F}_{2^m}^*)^3$ with $a_{11}=a_{22}$. Let us define the following sets: 
\begin{align*}
T &=\{(d_1,d_3) \neq (0,0): d_1 \neq d_3\}\\
\text{and}~~ X &=\{(d_1,d_3)\in T:a_{11}d_1+a_{33}d_3 \neq 0, a_{11}d_1+a_{22}d_2+a_{33}d_3=a_{11}d_1+a_{11}d_3+a_{33}d_3 \neq 0 \}
\end{align*}
Clearly, $|T|=(2^m-1)(2^m-2)$. 
To determine the cardinality of $X$, we first count the cardinality of the following sets.
Let 
\begin{align*}
X_1 &=\{(d_1,d_3) \in T: a_{11}d_1+a_{33}d_3 = 0\},\\
X_2 &=\{(d_1,d_3) \in T:a_{11}d_1+(a_{11}+a_{33})d_3 = 0\},\\
\text{and}~~ Y_1 &=\{(xa_{33},xa_{11}): x \in \mathbb{F}_{2^m}^*\},\\
Y_2 &=\{(z(a_{11}+a_{33}),z(a_{11}): z \in \mathbb{F}_{2^m}^*\}.
\end{align*}
Our claim is that $X_i=Y_i$ for $i=1,2$. One side of the inclusion is evident, i.e., $Y_1 \subseteq X_1,Y_2 \subseteq X_2$.  For the converse part,
consider an arbitrary element $(\alpha_1,\alpha_2 ) \in X_1$. Then $a_{11}\alpha_1+a_{33}\alpha_2 = 0$. There exists non-zero elements $\beta_1, \beta_2$ over $\mathbb{F}_{2^m}^*$ such that $\alpha_1=\beta_1a_{33}$ and $\alpha_2=\beta_2a_{11}$. Consequently, we get $a_{11}a_{33}(\beta_1+\beta_2)=0$ which implies $\beta_1=\beta_2$. Thus $X_1=Y_1$. Similarly, $X_2=Y_2$ and $|X_1|=|X_2|=(2^m-1)$. 

Next we calculate $|X_1\cup X_2|$.
If $X_1 \cap X_2 \neq \phi$  then $a_{33}=a_{11}+a_{33}$, which is contradiction to $a_{11} \neq 0$. Therefore
\begin{align*}
|X_1\cup X_2|&=2(2^m-1)\\
\text{and}~ |X|=|T| \setminus |X_1\cup X_2|&= (2^m-1)(2^m-2)-2(2^m-1)=(2^m-1)(2^m-4).
\end{align*}
Thus $|S_3''|=(2^m-1)^2(2^m-2)(2^m-4)$.

\textbf{Case III.} Consider the set $S_3'''$. Then $d_1 \neq d_2, d_2 \neq d_3, d_1=d_3$.

Proving similarly as Case II of Lemma \ref{3rd set}, we will get $|S_3'''|=(2^m-1)^2(2^m-2)(2^m-4)$. 

\textbf{Case IV.} Consider the set $S_3''''$. Then $d_1 \neq d_2, d_2 \neq d_3, d_1 \neq d_3$.

First fix an arbitrary triple $(a_{11},a_{22},a_{33}) \in (\mathbb{F}_{2^m}^*)^3$ with $a_{11}=a_{22}$.
% such that $a_{33}=a_{11}+a_{22}$. 
Let define the following sets:
\begin{align*}
T &=\{(d_1,d_2,d_3) \neq (0,0,0): d_i \neq d_j , 1 \leq i<j \leq 3\},\\
\text{and}~~ X&=\{(d_1,d_2,d_3) \in T: a_{11}d_1 + a_{22}d_2= a_{11}(d_1 + d_2)\neq 0, \\
&~~~~ a_{11}d_1 + a_{33}d_3 \neq 0, a_{22}d_2 + a_{33}d_3 \neq 0, a_{11}d_1 + a_{22}d_2+ a_{33}d_3 \neq 0\}.
\end{align*}
Clearly, $|T|= (2^m-1)(2^m-2)(2^m-3)$. To determine the cardinality of $X$, we begin with the following four subsets of $X$.  
\begin{align*}
X_1&=\{(d_1,d_2,d_3) \in T: a_{11}d_1 + a_{22}d_2 = 0\},\\
X_2&=\{(d_1,d_2,d_3) \in T: a_{11}d_1 + a_{33}d_3 = 0\},\\
X_3&=\{(d_1,d_2,d_3) \in T: a_{11}d_2 + a_{33}d_3 = 0\},\\
\text{and}~
X_4&=\{(d_1,d_2,d_3) \in T: a_{11}d_1 + a_{22}d_2+a_{33}d_3 = 0\}.
\end{align*}
Observe that $|X|=|T| \setminus|X_1\cup X_2\cup X_3\cup X_4|$.
Since $d_1 \neq d_2$ and $a_{11}=a_{22},~ a_{11}d_1 + a_{22}d_2= a_{11}(d_1 + d_2)$ is always non-zero  for all $(d_1,d_2,d_3)$. Therefore $|X_1|= 0$. 

Consider the following three sets: 
\begin{align*}
Y_2 &=\{(ya_{33},d_2,ya_{11}): y \in \mathbb{F}_{2^3}^*,d_2\neq \{0, ya_{33}, ya_{11}\}\},\\
Y_3 &=\{(d_1,za_{33},za_{11}): d_1 \in \mathbb{F}_{2^3}^*, z \in \mathbb{F}_{2^3}^*,d_1 \neq \{za_{33}, za_{11}\}\}\\
&= \{(d_1,za_{33},za_{11}): d_1 \in \mathbb{F}_{2^3}^*, z \in \mathbb{F}_{2^3}^*, z \neq \{ d_1a_{33}^{-1}, d_1a_{22}^{-1}\} \},\\
\text{and}~~Y_4 &=\{(d_1,d_2,d_3) : d_1 \in \mathbb{F}_{2^3}^*, d_2 \neq \{0, d_1, (a_{11}+a_{33})a_{22}^{-1}d_1, (a_{22}+a_{33})^{-1}a_{11}d_1\}, d_3 \neq 0,\\
&~~ d_3 = a_{33}^{-1}(a_{11}d_1 + a_{22}d_2)\}.
\end{align*}

Using the same argument as previous cases, we have $Y_2=X_2, Y_3=X_3$. Then $|X_2|=|X_3|=(2^m-1)(2^m-3)$. 

Next we prove that $Y_4=X_4$.
%for the case $a_{33} \neq a_{11}+a_{22}$.\\
Let $(d_1',d_2',d_3')\in Y_4$. Then $ d_3'=a_{33}^{-1}(a_{11}d_1' + a_{22}d_2') $ implies $ a_{11}d_1' + a_{22}d_2'+a_{33}d_3'=0 $. Additionally, since $d_3' \neq 0$ it follows that $a_{11}d_1' \neq a_{22}d_2'$ i.e., $d_2' \neq a_{22}^{-1}a_{11}d_1'=d_1'$. To prove $(d_1',d_2',d_3')\in X_4$, we need to show  $d_3' \neq 
\{d_1',d_2'\}$. If we assume $d_3'= d_1'$ then $d_1' = a_{33}^{-1}(a_{11}d_1' + a_{22}d_2')$ which implies $d_2'=(a_{11}+a_{33})a_{22}^{-1}d_1'$, which is not possible. Also, if $d_3'=d_2',$ then $d_2' = a_{33}^{-1}(a_{11}d_1' + a_{22}d_2')$ implies  $d_2'=(a_{22}+a_{33})^{-1}a_{11}d_1'$, also not possible. Therefore $(d_1',d_2',d_3') \in T$ and $Y_4 \subseteq X_4$.

Conversely , let $(d_1',d_2',d_3')\in X_4$. Then $d_1' \neq d_2', d_1'\neq d_3', d_1' \neq d_2'$ and $a_{11}d_1' + a_{22}d_2' + a_{33}d_3'=0$. Since $ a_{33}d_3' \neq 0$, we have $a_{11}d_1' \neq a_{22}d_2'$. This implies $d_2' \neq a_{22}^{-1}a_{11}d_1'=d_1'$. Since $d_3'\neq \{d_1', d_2'\}$, then $a_{11}d_1' + a_{22}d_2' \neq a_{33}d_1'$ and $a_{11}d_1' + a_{22}d_2' \neq a_{33}d_2'$. This implies $d_2' \neq (a_{11}+a_{33})a_{22}^{-1}d_1'$ and $d_2' \neq (a_{22}+a_{33})^{-1}a_{11}d_1'$ respectively. Thus $(d_1',d_2',d_3')\in Y_4.$ Thus $|Y_4|=(2^m-1)(2^m-4)=|X_4|$.

Observe that $X_1 \cap X_2= X_1 \cap X_3= X_1\cap X_4 = \phi$. Now we calculate cardinality of $X_2 \cap X_3$. If $(d_1',d_2',d_3') \in X_2 \cap X_3$, then $a_{11}d_1' + a_{33}d_3' = 0$ and $a_{11}d_2' + a_{33}d_3' = 0$. Adding these two equations, we obtain, $a_{11}d_1'+a_{11}d_2'=0$ which is not possible since $d_1' \neq d_2'.$ Therefore we conclude that $X_2 \cap X_3 =\phi$. 

Similarly, for the set  $X_2 \cap X_4$, assume that $(d_1',d_2',d_3') \in X_2 \cap X_4$. Then $a_{11}d_1' + a_{33}d_3' = 0$ and $a_{11}d_1' + a_{22}d_2'+a_{33}d_3' = 0$. These two equations imply $a_{33}d_3' = 0$ which is not possible. For similar reasons, $X_3  \cap X_4$ is also empty.
Therefore, 
\begin{align*}
|X_1\cup X_2\cup X_3\cup X_4|&=|X_1|+|X_2|+|X_3|+|X_4|\\
&=2(2^m-1)(2^m-3)+(2^m-1)(2^m-4)\\
&=(2^m-1)(3\cdot2^m-10)\\
\text{and}~~|X|&=(2^m-1)(2^m-2)(2^m-3)-(2^m-1)(3\cdot2^m-10)\\
&=(2^m-1)(2^{2m}-8\cdot2^m+16).
\end{align*}
Hence cardinality of $S_3''''$ is $(2^m-1)^2(2^m-2)(2^{2m}-8\cdot2^m+16)$. Since the intersection of any two of $S_3',S_3'',S_3'''$ and $S_3''''$ are empty, cardinality of $S_3$ is
\begin{align*}
|S_3|&=2(2^m-1)^2(2^m-2)(2^m-4)+(2^m-1)^2(2^m-2)(2^{2m}-8\cdot2^m+16)\\
&=(2^m-1)^2(2^m-2)(2^{2m}-6 \cdot2^m+8).
\end{align*}
\end{proof}

Consider the other two cases similar to $S_3$ and named them as follows: 
\begin{align*}
S_4&= \{(a_{11}, a_{22}, a_{33}, d_1, d_2, d_3) \in (\mathbb{F}_{2^m}^*)^6, m \geq 2: a_{11}=a_{33}, a_{22} \neq a_{33}, a_{11}d_1 + a_{22}d_2 \neq 0,\\
&~~  a_{11}d_1 +a_{33}d_3 \neq 0, a_{22}d_2 + a_{33}d_3 \neq 0, a_{11}d_1 + a_{22}d_2 + a_{33}d_3 \neq 0\}\\
\text{and}~ S_5 &= \{(a_{11}, a_{22}, a_{33}, d_1, d_2, d_3) \in (\mathbb{F}_{2^m}^*)^6, m \geq 2: a_{22}=a_{33}, a_{11} \neq a_{33}, a_{11}d_1 + a_{22}d_2 \neq 0,\\
&~~  a_{11}d_1 +a_{33}d_3 \neq 0, a_{22}d_2 + a_{33}d_3 \neq 0, a_{11}d_1 + a_{22}d_2 + a_{33}d_3 \neq 0\}.
\end{align*}
Then, we have 
\begin{align*}
|S_3 \cup S_4 \cup S_5|=& 6(2^m-1)^2(2^m-2)(2^m-4)+3(2^m-1)^2(2^m-2)(2^{2m}-8\cdot2^m+16)\\
&=3(2^m-1)^2(2^m-2)(2^{2m}-6\cdot2^m+8)\\
&=3(2^m-1)^2(2^m-2)^2(2^m-4).
\end{align*}

%\begin{theorem}\label{cardinality final}
%Let $S$ be the set $S= \{(a_{11}, a_{22}, a_{33}, d_1, d_2, d_3) \in (\mathbb{F}_{2^m}^*)^6, m \geq 2: a_{11}d_1 + a_{22}d_2 \neq 0, a_{11}d_1 +a_{33}d_3 \neq 0, a_{22}d_2 + a_{33}d_3 \neq 0, a_{11}d_1 + a_{22}d_2 + a_{33}d_3 \neq 0\}$. Then cardinality of $S$ is $(2^m-1)^2(2^m-2)(2^{3m}-6\cdot2^{2m}+9\cdot2^m-4)$.
%\end{theorem}

\begin{theorem}\label{cardinality final}
Let $S$ be the set $S= \{(a_{11}, a_{22}, a_{33}, d_1, d_2, d_3) \in (\mathbb{F}_{2^m}^*)^6, m \geq 2: a_{11}d_1 + a_{22}d_2 \neq 0, a_{11}d_1 +a_{33}d_3 \neq 0, a_{22}d_2 + a_{33}d_3 \neq 0, a_{11}d_1 + a_{22}d_2 + a_{33}d_3 \neq 0\}$. Then cardinality of $S$ is $(2^m-1)^3(2^m-2)(2^{m}-4)$.
\end{theorem}
\begin{proof}
Note that, the set $S$ is the disjoint union of $S_1,S_2,S_3,S_4$ and $S_5$ i.e.,  $S_i \cap S_j= \phi$ for all $1 \leq i <j \leq 5$. Therefore, from lemma \ref{1st set}, \ref{2nd set} and \ref{3rd set}, we have $|S|=|S_1|+|S_2|+|S_3|+|S_4|+|S_5|=(2^m-1)^3(2^m-2)(2^{m}-4)$.
\end{proof}

%\begin{theorem}\label{counting result}
%The number of $3 \times 3$ semi-involutory MDS matrix over the finite field $\mathbb{F}_{2^m}, m \geq 2$ is $ (2^m-1)^4(2^m-2)(2^{3m}-6\cdot2^{2m}+9\cdot2^m-4)$.
%\end{theorem}

\begin{theorem}\label{counting result}
The number of $3 \times 3$ semi-involutory MDS matrix over the finite field $\mathbb{F}_{2^m}, m \geq 2$ is $ (2^m-1)^5(2^m-2)(2^{m}-4)$.
\end{theorem}
\begin{proof}
An MDS semi-involutory matrix is expressed in general form given by the equation (\ref{matrix1}) using only diagonal entries of the matrix and the entries of an associated diagonal matrix. Let $a_{11}, a_{22}, a_{33}, d_1, d_2, d_3$ are those entries and $x,y$ are arbitrary. Since $x, y \in \mathbb{F}_{2^m}^*$, then the number of choices for $x$ and $y$ is $(2^m-1)^2$. Using Theorem \ref{SI MDS} and Theorem \ref{cardinality final}, the number of choices for $a_{11}, a_{22}, a_{33}, d_1, d_2$ and $d_3$ are $(2^m-1)^3(2^m-2)(2^{m}-4)$. Therefore total number of MDS semi-involutory matrix is  $ (2^m-1)^5(2^m-2)(2^{m}-4)$.
\end{proof}

\begin{remark}
In \cite{GSARC} it was proved that the number of $3 \times 3$ involutory MDS matrices over $\mathbb{F}_{2^m}, m >2$ is $(2^m-1)^2(2^m-2)(2^m-4)$. Therefore, there exists total $1176$ and $37800,~ 3 \times 3$ involutory MDS matrices over $\mathbb{F}_{2^3}$ and $\mathbb{F}_{2^4}$ respectively. However, Theorem \ref{counting result} states that there does not exist any $3 \times 3$ semi-involutory MDS matrix over $\mathbb{F}_{2^2}$ and the number of $3 \times 3$ semi-involutory MDS matrices over $\mathbb{F}_{2^3}$ is $403368$, and over $\mathbb{F}_{2^4}$ is $127575000$.  
\end{remark}

\section{Conclusion}\label{conclusion}
We prove that it is possible to construct all $3\times 3$ semi-involutory MDS matrices over the finite field $\mathbb{F}_{2^m}$ by using only three diagonal elements and the entries of an associated diagonal matrix. Our proposed matrix form in equation (\ref{matrix1}) with MDS property is beneficial to use in the diffusion layer of SPN based block cipher because the inverse matrix is a simple matrix multiplication of two diagonal matrices with the original matrix. Additionally, we provide the count of the total number of $3 \times 3$ MDS, semi-involutory  matrix over the finite field $\mathbb{F}_{2^m}$.

The general structures for the involutory and semi-involutory  matrices of even sizes
or powers of $2$ is still an open problem. Also it will be worth exploring when those matrices have MDS properties.

%\begin{section}{Acknowledgement}
%The authors thank all the reviewers especially reviewer three for the important suggestion which improve the bound of the Theorem \ref{counting result}. 
%The second author thanks Susanta Samanta for his helpful suggestion on an earlier version of this paper.
%
%\end{section}

\section{Declaration}
\textbf{Competing interest:} The authors declare that they have no known competing financial interests or personal
relationships that could have appeared to influence the work reported in this paper.

\bibliographystyle{plain}

\end{document}